\documentclass[11pt]{article}
\usepackage{amsmath,amssymb,amsthm,amsfonts,latexsym,bbm,xspace,graphicx,float,mathtools,mathdots,cancel}
\usepackage{braket,caption,subcaption,ellipsis,xcolor,textcomp,MnSymbol}
\usepackage[colorlinks,citecolor=blue,bookmarks=true]{hyperref}
\usepackage[nameinlink]{cleveref}
\crefname{ineq}{inequality}{inequalities}
\usepackage{url}
\creflabelformat{ineq}{#2{\upshape(#1)}#3}
\usepackage[letterpaper,margin=1in]{geometry}
\usepackage{enumitem}
\usepackage[ruled,vlined,linesnumbered]{algorithm2e}

\usepackage{atbegshi}
\usepackage{bibentry} 
\usepackage{authblk}
\usepackage[explicit]{titlesec}
\newtheorem{theorem}{Theorem}[section]
\newtheorem{lemma}[theorem]{Lemma}

\newtheorem{definition}[theorem]{Definition}

\newtheorem{problem}[theorem]{Problem}

\newcommand\N{{\mathbb{N}}}

\newcommand\poly{\mathrm{poly}}
\newcommand\remove[1]{{}}
\newcommand\sF{{\mathcal{F}}}

\crefname{question}{question}{questions}

\title{Breaking the $2^n$ barrier for graph $k$-coloring}
\author{Kevin Pratt\thanks{ktp2116@columbia.edu}}
\affil{Columbia University}
\date{}
\begin{document}

\maketitle
\begin{abstract}
	We show that for all $k$, there exists $\varepsilon_k > 0$ such that graph $k$-coloring can be solved by a randomized algorithm with one-sided error in time $O((2-\varepsilon_k)^n)$. Prior to this work and independent concurrent work of
	Zamir~\cite{zamirnew}, exponential improvements over the $2^n \cdot \poly(n)$-time algorithm of Bj\"orklund, Husfeldt, and Koivisto \cite{bjorklund2009set} were only known for $k \le 6$.
\end{abstract}

\section{Introduction}
In this paper, we study algorithms for the graph $k$-coloring problem:
\begin{problem}[\textsc{Graph $k$-coloring}]
	Given as input a graph $G$ on $n$ vertices, decide if it is $k$-colorable, that is, if there exists a coloring of its vertices with $k$ colors such that no two adjacent vertices are assigned the same color.
\end{problem}
Graph $k$-coloring for $k \ge 3$ is one of the canonical NP-complete problems \cite{garey2002computers, karp2009reducibility}. Yet despite its intractability, researchers have sought out faster and faster exponential-time algorithms for this basic problem. As a baseline, it can be na\"ively solved in time $O^*(k^n)$\footnote{Throughout, $O^*(-)$ suppresses factors of $\poly(n)$.} by enumerating all $k$-colorings and checking their validity. The state-of-the-art is the following: outside of speedups for small $k$ and for structured graphs (which we review in the next subsection), the fastest algorithm is due to Bj\"orklund, Husfeldt, and Koivisto \cite{bjorklund2009set}, and runs in time $O^*(2^n)$.

This situation contrasts with that for a large number of other NP-complete problems naturally indexed by a parameter $k$, (most notably $k$-SAT \cite{paturi2005improved}, and more generally, arity-$k$ boolean CSPs \cite{schoning1999probabilistic}), where for any fixed $k$, current algorithms are exponentially faster than for the unbounded-$k$ problem. That is to say, they have runtimes of the flavor $(2-\varepsilon_k)^n$ where $\varepsilon_k > 0$ and $\lim_{k \to \infty} \varepsilon_k = 0$.

We establish such a speedup for graph coloring:

\begin{theorem}\label{thm:fin}
For all $k$, there exists $\varepsilon_k > 0$ such that \textsc{Graph $k$-coloring} can be solved in time $O((2-\varepsilon_k)^n)$. The algorithm is randomized with one-sided error: it accepts on a yes instance with probability at least $2/3$, and rejects on any no instance.
\end{theorem}
Independent and concurrent work of Zamir \cite{zamirnew}, posted shortly before this manuscript, proves the same main theorem by a somewhat different approach.
\subsection{Prior work}
In 1976, Lawler gave the first non-trivial algorithm for graph coloring, with runtime $O(2.45^n)$ \cite{lawler1976note}. This algorithm was based on two observations: first, a $k$-coloring is a partition of $V(G)$ into $k$ independent sets, and second, we may assume one of these is maximal. By combining these facts with a fast algorithm for enumerating maximal independent sets, the result followed. This was subsequently improved by Eppstein to $O(2.42^n)$ \cite{eppstein2003small}, and later by Byskov to $O(2.41^n)$ \cite{byskov2004enumerating}.

The next milestone was the $O^*(2^n)$-time algorithm of Bj\"orklund, Husfeldt, and Koivisto \cite{bjorklund2009set}. Up to this point, coloring algorithms were combinatorial. Their algorithm deviated from that by combining the independent-set partitioning characterization of colorability with an algebraic component: an algorithm for fast subset convolution. In fact, their algorithm follows as a corollary of fast subset convolution.

Subsequent work then led to truly sub-$2^n$-time algorithms in two scenarios: for small $k$, and for graphs with structure. 

\paragraph{Small $k$.}
In the case when $k=3$, the approach of Lawler in fact runs in time $O^*(3^{n/3}) < O(1.45^n)$. This was subsequently improved to $O(1.42^n)$ by Schiermeyer \cite{schiermeyer1993deciding}, and later to $O(1.33^n)$ by Beigel and Eppstein \cite{beigel20053}.

For $k=4$ coloring, a $O(1.81^n)$ time algorithm was given in \cite{beigel19953}. This was improved to $O(1.728^n)$ in \cite{fomin2007improved}, and then to $O(1.716^n)$ in \cite{wu2024faster}.

For $k=5$ and $6$, Zamir gave an exponential speedup by combining the approach of \cite{bjorklund2009set} with new combinatorial insights, and in particular, with faster algorithms for certain structured graphs \cite{zamir2021breaking}.

\paragraph{Structured graphs.}
In \cite{bjorklund2008trimmed} it was shown that for graphs of maximum degree $\Delta$, $k$-coloring can be solved in time $O((2-\varepsilon_{\Delta}))^n$. This was based on an improved algorithm for subset convolution in the case when the sets in the input families have small cardinality (for instance, bounded by $n/10$).

In \cite{golovnev2016families}, this was subsequently generalized to the class of graphs of average degree at most $\Delta$. That was generalized even further by a $(2-\varepsilon_{\alpha,\Delta})$-time algorithm for the case when $\alpha$n vertices have degree at most $\Delta$ \cite{zamir2021breaking}.

In a different direction, in \cite{zamir2023algorithmic} it was shown that for every fixed
$k$ and $C$, $k$-\textsc{Coloring} admits an
$O((2-\varepsilon_{k,C})^n)$-time algorithm on graphs satisfying
\[
\Delta_{\max}(G)\leq C\,d_{\mathrm{avg}}(G).
\]
This includes regular graphs of arbitrary degree and graphs of
any fixed positive density.
\mbox{}\\

These results cover broad regimes, including sparse, regular, and dense graphs. They leave open, however, the problem for arbitrary highly irregular graphs, for any $k \ge 7$.

\section{Our approach}
We begin by focusing on the problem $k$-coloring with specified color class sizes $(n_1, n_2, \ldots, n_k)$, with $\sum n_i = n$. We will refer to $p := (n_1/n, \ldots, n_k/n)$ as the \emph{profile} of the coloring we wish to detect. Note that it suffices to design an algorithm for this problem, since for fixed $k$ the cost of enumerating all $\binom{n+k-1}{k-1}$ profiles is negligible.

A key quantity throughout the proof is the mass $p(R)$ of a set of colors (or \emph{palette}) $R$. We say that a set of vertices $A$ is \emph{$R$-realizable} if $A$ can be partitioned into $|R|$ independent sets of sizes $(p_i \cdot n)_{i \in R}$. We say that $R$ is $\eta$-good if $p(R) \le 1/2 - \eta$. Such sets are key to our approach; their significance is that for an $\eta$-good $R$, we can construct all $R$-realizable subsets of $V(G)$ in time $2^{H(1/2 - \eta)n} \ll 2^n$ using the trimmed M\"obius inversion algorithm of \cite{bjorklund2008trimmed}. At a very high level, to improve on the $2^n$ algorithm, we will find a way to avoid having to na\"ively construct $R$-realizable subsets of vertices when $p(R) > 1/2 - \eta$. Our main trick is to only perform ``above-$(1/2 - \eta)$" subset convolutions on significantly restricted universes.

\paragraph{Speedup mechanism}
The basic mechanism behind our result is a speedup for $k$-coloring with certain list coloring constraints (\Cref{lem:forcing}). Suppose we wish to detect $k$-colorings with profile $p$, with the further restriction that some set of vertices $A$ have color contained in $R \subset [k]$. Then if either $|R| \le 2$ or $R$ is $\eta$-good, we show how to solve this in time $O^*(2^{H(1/2 - \eta)n} + 2^{n-|A|})$. Moreover, the cost of $2^{H(1/2 - \eta)n}$ can be paid up-front as a preprocessing step. 

The algorithm is simple: to preprocess, precompute all pairs $(R,B)$ for $R \subset [k]$ with $R$ $\eta$-good and $B \subset V(G)$ such that $B$ is $R$-realizable. This can be done in time $O^*(2^{H(1/2 - \eta)n})$ using fast trimmed subset convolution. It is then relatively straightforward to reduce a query $A \subseteq V(G), R \subset [k]$ to doing a subset convolution on the universe $V(G) \setminus A$. 

We will use this primitive repeatedly in what follows.
\paragraph{Our algorithm}
Our main algorithm (\Cref{thm:main}) is a combination of three different algorithms, which will detect different types of witness colorings. This algorithm will have the property that it will reject on any non-$k$-colorable graph, and accept with constant probability on any graph that is $k$-colorable with profile $p$. It is therefore allowed to accept if $G$ is $k$-colorable with a profile other than $p$, which suffices as ultimately we aim to solve $k$-colorability (and not $k$-coloring with profile $p$).
\paragraph{Many low-degree vertices}
If at least 1\% of vertices in $G$ have degree at most some $d = d(k)$ (one can imagine $d \approx 8^k$), we apply the algorithm for many bounded-degree vertices of \cite{zamir2021breaking}. Since $d$ is a function of $k$ and 1\% is a constant, this runs in $\ll 2^n$ time.
\newline\mbox{}

The remaining hard case is now when 99\% of vertices have degree at least $d$. The rest of our approach is based on the following dichotomy for the case. Fix a witness coloring $C = (C_1, \ldots, C_k)$ and a threshold $T = T(k)$. For a vertex $v$, we then let 
\[\Sigma_T(v) = \{i \in [k]: |N(v) \cap C_i| \ge T\}\]
and let
\[U_{R} = \{v \in V(G) : \Sigma_T(v) = R\}.\]
In words, $U_{R}$ is the set of vertices which, in the witness coloring, see each color in $R$ at least $T$ times, and see every other color fewer than $T$ times. 

The sets $U_R$ are unknown by the algorithm and are only used in the analysis. We later take $d \ge T \cdot (k-1) + 1$, so that the sets $U_R$ partition the degree at least $d$ vertices. In particular, 99\% of vertices are contained in some $U_R$.

Then, either:

\begin{enumerate}
	\item At least a $15/16$ fraction of vertices are contained in $\bigcup_{R\text{ $\eta$-good}} U_R$;
	\item At least a $99/100 - 15/16 > 1/32$ fraction of vertices are contained in $\bigcup_{R : p(R) > 1/2 - \eta} U_R$, and hence at least a $2^{-k-5}$ fraction of vertices are contained in some particular $U_{R_0}$ with $p(R_0) > 1/2 - \eta$.
\end{enumerate}
These cases exhaust all possible remaining witness colorings. We give two algorithms which give speedups for these cases.

\paragraph{Case 1: Many vertices almost see $\eta$-good palettes}

To detect such a witness, we use a branching algorithm (\Cref{lem:degen}). For each $\eta$-good set $R$, this algorithm will maintain sets $W_R$ of vertices assigned to $\eta$-good palettes $R$. Initially these are empty. As the main step of the algorithm, we select an unassigned vertex of degree at least $d$ among the unassigned vertices at random. Then we branch: for each $\eta$-good palette, and all subsets of this vertices first $d$ neighbors of size at least $d-Tk$, we assign these neighbors to this palette. Note also that for one of the ``good" vertices in an $\eta$-good $R_0$, fewer than $Tk$ of its neighbors do not have witness color in $R$ by definition of $U_R$, so some branch will produce an assignment to the neighborhood consistent with the witness.

This is repeated until one of two conditions holds. The first is that for some $R$, $|W_R| > \lambda n$, for some small $\lambda(k)$. In this case, we run \Cref{lem:forcing}. If this condition does not hold, but if at least half of the vertices in the subgraph induced by the unassigned vertices have degree at most $d$, then we brute-force the colors of the assigned vertices $\cup_{R} W_R$ (of which there are fewer than $2^k \lambda n$), and then use the many low-degree list coloring algorithm of \cite{zamir2021breaking} on the unassigned vertices.


\paragraph{Case 2: Many vertices see a non-$\eta$-good palette}
If Case 1 fails, there exists some $R_0$ that is \emph{not} $\eta$-good and where $|U_{R_0}| > 2^{-k-5} \cdot n$. Our algorithm in this case is based on the following natural idea: sample a set $S$ of size $\sigma n$ for $\sigma= \ln(4k)/T$, and brute-force its $k$-colorings. If $T$ is large enough, we show that this will create ``probabilistic" list-coloring restrictions on a large enough set of vertices that the cost of the brute-force is paid for by the savings from \Cref{lem:forcing}.

More precisely, define $V_R(S)$ to be the vertices in $G$ which see exactly the colors in $R$ in $S$. Note that these sets are easily determined algorithmically after guessing a coloring of $S$. Using that every vertex $v \in U_{R_0}$ sees all colors in $R_0$ at least $T$ times, standard probabilistic considerations plus pigeonholing show that, with constant probability, there exists an $R$ with $R_0 \subseteq R$ and $|V_R| \ge 2^{-2k-6} \cdot n$.

Because $R_0 \subseteq R$, $p(R) \ge p(R_0) > 1/2 - \eta$. If we could upgrade this very slightly to $p(R) > 1/2 + \eta$, we would be done, because then we could directly apply \Cref{lem:forcing} to assign $V_R$ the palette $R^c$, which would satisfy $p(R^c) < 1/2 - \eta$. Ignoring preprocessing, this would give a runtime of roughly
\[k^{\sigma n} \cdot 2^{n(1-4^{-k})} =k^{\ln(4k) n/T} \cdot 2^{n(1-4^{-k})} \ll 2^n\]
for $T$ a sufficiently large function of $k$.

Such an upgrade need not hold, so instead we use an averaging argument to show that there is some $R' \subseteq R^c$ such that either $|R'| \le 2$ or $R'$ is $\eta$-good, and whose colors make up significantly more than half of $V_R$ in the witness (\Cref{lem:averaging}). More formally, let $u_R$ be the distribution of colors in the witness in $V_R$. Then if $\eta < 1/(6(k-1))$ there exists $R' \subset R^c$ that is either $\eta$-good or where $|R'| \le 2$, and such that $u_R(R') > 1/2 + 1/(6k)$. In either case, we obtain a palette we can use in \Cref{lem:forcing}.

The upshot is the following. Suppose we sample a subset $A$ of $V_R$ of size $|V_R|/(6k)$. Then, a calculation shows that all vertices in $A$ have witness color in $R'$ with probability $2^{-\zeta |A|}$ for some $\zeta = \zeta(k) < 1$. We then run \Cref{lem:forcing}, and repeat this procedure $2^{\zeta |A|}$ times to succeed with constant probability. The final runtime is then roughly
\[k^{\sigma n} \cdot 2^{\zeta |A|} \cdot 2^{n-|A|} = k^{\ln(4k)n/T} 2^{n(1-f(k))} \ll 2^n\]
after taking $T$ sufficiently large.




\section{Preliminaries}
Throughout the paper, $\log$s are base 2. For $x \in (0,1)$, $H(x) = -x \log x -(1-x) \log (1-x)$ is the binary entropy of $x$. We assume $k \ge 3$.

Given set families $\sF_1, \sF_2 \subseteq 2^{[n]}$, let $\sF_1 \ostar \sF_2$ denote their subset convolution:
\[\sF_1 \ostar \sF_2 := \{X \subset [n] : X = A \sqcup B \text{ for } A \in \sF_1, B \in \sF_2\}\]

A key primitive is the following algorithm for trimmed subset convolution.
\begin{theorem}[\cite{bjorklund2008trimmed}] \label{lem:construct}
	Given set families $\sF_1, \sF_2 \subseteq 2^{[n]}$, consisting of sets of cardinality bounded by $a$ and $b$ respectively, we can compute $\sF_1 \ostar \sF_2$ in time $O^*(\binom{n}{\le a+b})$.
\end{theorem}

Following \cite{zamir2021breaking}, we call a graph $(\alpha, \Delta)$-bounded if $\alpha n$ vertices have degree at most $\Delta$.
\begin{theorem}[Theorem 5, \cite{zamir2021breaking}]\label{thm:zamir}
	For every fixed $k,\alpha,\Delta$, there exists
	$\mu_{\alpha,\Delta,k}>0$ such that $k$-\textsc{List Coloring}
	on an $(\alpha,\Delta)$-bounded $n$-vertex graph can be solved
	in time
	\[
	O((2-\mu_{\alpha,\Delta,k})^n).
	\]
\end{theorem}

We will be interested in deciding if $G$ is colorable with a given \emph{profile}/color distribution $p =  (p_1, \ldots, p_k)$, where we're allowed $p_i n$ vertices of color $i$.

\begin{definition}
	For a profile $p$ and $R \subseteq [k]$ let $p(R) = \sum_{i \in R} p(i)$. We call $R \subseteq [k]$ $\eta$-good if $p(R) \le 1/2 - \eta$.
\end{definition}
\section{Proof}

\begin{lemma}\label{lem:forcing}
Given as input an $n$-vertex graph $G$ and profile $p$ and $\eta > 0$, we can preprocess $G$ in time $O^*(2^{H(1/2 - \eta)n})$ to solve the following queries in time $O^*(2^{n - |A|})$:

Given as input $A \subseteq V(G)$ and $R \subseteq [k]$ such that either $p(R) \le 1/2 - \eta$ or $|R| \le 2$, decide if $G$ admits a $k$-coloring with profile $p$, and where all $v \in A$ have color in $R$.

\end{lemma}
\begin{proof}
We begin by enumerating all independent sets of size $p_i \cdot n$, whenever $p_i \le 1/2 - \eta$. Let $\sF_i$ be the resulting set family. For every $\eta$-good $R$, use \Cref{lem:construct} to compute 
\[\sF_R := \bigostar_{c \in R} \sF_c .\]
Because $R$ is $\eta$-good, all $\sF_i$ are defined. This is done in time $2^{H(1/2 - \eta)n} \cdot \poly(n)$. For all $S \in \sF_R$, we insert the tuple $(R,S)$ into a hash table $H$. This concludes the preprocessing step.

Now, given a query $A \subset V(G),  R \subset [k]$:
\begin{itemize}
	\item Suppose $R$ is $\eta$-good. For every $S \subseteq V(G) \setminus A$, check if $(R, S \cup A) \in H$, and if it is, add $S$ to a new set family $\sF_R'$. This takes time $2^{n - |A|} \cdot \poly(n)$. Next, for all $i \notin R$, construct the set family $\sF_i'$ consisting of independent sets $S \subseteq V(G) \setminus A$ of size $p_i n$. This is done by brute force in time $O^*(2^{n-|A|})$. Finally, using \Cref{lem:construct}, compute
	\[\sF= \sF_R' \ostar \left (\bigostar_{i \notin R} \sF_i' \right ).\]
	Because all sets are contained in $V(G) \setminus A$, this takes time $2^{n - |A|}$. It is readily observed that $V(G) \setminus A \in \sF$ if and only if a coloring satisfying the desired constraints exists. 
	
	\item Suppose $R = \{i\}$. Enumerate all $B \subset V(G) \setminus A$ where $A \cup B$ is an independent set and $|A|+|B| = p_i n$; call this family $\sF_i'$. This takes time $O^*(2^{n - |A|})$. Also construct all $\sF_j'$ of independent sets on $V(G) \setminus A$ of size $p_j n$, for $j \neq i$; this also takes time $O^*(2^{n-|A|})$. Compute $\bigostar_{i \in [k]} \sF_i'$, and return yes iff $V(G) \setminus A$ is in the output.

	\item Suppose $R = \{i,j\}$. Enumerate all subsets $B$ of $V(G) \setminus A$ of size $p(R)n - |A|$ such that $A \cup B$ is 2-colorable, and moreover admits a 2-coloring with color class sizes $p_i n$ and $p_j n$; for each $B$ this is done in polynomial time\footnote{Find all bipartite components of $G[A \cup B]$. For each one record pairs $(a_{k,0}, a_{k,1})$ of bipartition sizes, and then run a standard dynamic program to see if there exists $v \in \{0,1\}^m$ such that $p_i n = \sum_{k=1}^m a_{k,v_i}$.} so in total this can be done in time $O^*(2^{n - |A|})$. Call this $\sF'_{ij}$. For all  $\ell \notin \{i,j\}$ construct the family $\sF_\ell'$ of independent sets on $V(G) \setminus A$ of size $p_\ell n$; this also takes time $O^*(2^{n-|A|})$. Compute $\sF_{ij}' \ostar \bigostar_{\ell \notin \{i,j\}} \sF_\ell'$, and return yes iff $V(G) \setminus A$ is in the output.\qedhere
\end{itemize}

\end{proof}
\begin{lemma}[Speedup for many mostly-light neighborhoods]\label{lem:degen}
	Fix $k,c \in \N$, and $\eta > 0$.
	There exists $d(k,c) \in \N$ and $\varepsilon(k,c,\eta) > 0$ such that the following can be solved in time $O((2-\varepsilon)^n)$.
	
	Given as input a graph $G$ and a profile $p$, we are promised that if $G$ is $k$ colorable with profile $p$, then it admits a coloring with at least a $15/16$-fraction of vertices $v$ satisfying
	\begin{enumerate}
		\item $\deg(v) \ge d$,
		\item all but at most $c$ neighbors of $v$ are contained in some $\eta$-good palette.
	\end{enumerate}
	The algorithm outputs yes with constant probability if $G$ is $k$-colorable with profile $p$, and always outputs no if $G$ is not $k$-colorable.
\end{lemma}
\begin{proof}
We assume $n \ge n_0(k)$, where $n_0(k)$ is sufficiently large for the rounding inequalities below. The finitely many smaller instances are solved by exhaustive search.
\paragraph{Parameters}
We will prove this with
\[d = 2^{k+4}(c+1)(k+c+4).\]

Let $r := d-c,$ and $\mu := 1- \log (2-\mu_{1/2,d,k}) > 0$, with $\mu_{1/2, d,k} > 0$ from \Cref{thm:zamir}. We then let
\[\lambda = \min \left (2^{-k}/8,  \frac{\mu}{2^{k+1} (\log (2^{k+1} \cdot \sum_{i=0}^c \binom{d}{i})/(d-c) + \log k + \mu - 1)} \right ).\]

We begin with the preprocessing step of \Cref{lem:forcing}. Then we will run the following recursive subroutine $C_{rep} =  2^{\lceil 2^k \lambda n/(d-c)\rceil} $ times:

\begin{quote}
Throughout the subroutine's execution, we will maintain assignments of vertices to $\eta$-good palettes. We write $W_R$ for the vertices assigned to the palette $R$, $W := \sqcup_{R \text{ $\eta$-good}} W_R$, with $w := |W|$. Let $U = V(G) \setminus W$ be the unassigned vertices. Initially $W = \emptyset$.
\paragraph{Base case 1:} Suppose $|W_R| \ge \lambda n$ for some $R$. Use \Cref{lem:forcing} to solve the coloring problem with the constraints that $C(W_R) \subset R$, and accept if this succeeds. This takes time
\[O^*(2^{n - |W_R|}).\]

\paragraph{Base case 2:} Suppose $|W_R| < \lambda n$ for all $R$, but at least half of $U$ has degree at most $d$ in $G[U]$. Enumerate all $k^{w}$ colorings of $W$. For each valid coloring, we are left with a list coloring instance on $U$, with half of the vertices having degree at most $d$. Run the algorithm of \Cref{thm:zamir}, and accept if this succeeds. This takes time
\[O^*(k^w \cdot 2^{(1 - \mu)(n-w)}).\]

\paragraph{Main step:} 
If neither base case holds, pick a random vertex $u \in U$ with degree in $G[U]$ at least $d$. Consider the first $d$ neighbors of $u$ in some fixed ordering. Enumerate all subsets $X$ of these $d$ neighbors of size in the range $[d-c, d]$, as well as all nonempty $\eta$-good palettes $R$. Assign all $v \in X$ to $W_R$ and recurse.
 \end{quote}
If none of the $C_{rep}$ trials accepted, reject.
\paragraph{Correctness}
First consider when $G$ is $k$-colorable with profile $p$, with at least a $15/16$ fraction of vertices satisfying the degree and neighborhood conditions. We claim that this algorithm accepts with constant probability.

First, if we reach a main step, then neither base case holds. That is, more than half of $U$ has degree at least $d$ in $G[U]$ and $|U| \ge n - 2^k \lambda n \ge 7n/8$ (as $\lambda \le 2^{-k}/8$). So $U$ contains at least $7n/16$ vertices with degree at least $d$ in $G[U]$. Among these, at most $n/16$ can fail condition (2). Thus the probability that a randomly chosen $u$ with $\deg_U(u) \ge d$ satisfies the promised conditions is at least
\begin{equation}\label{eqn:uprob}
\frac{7n/16 - n/16}{7n/16} = 6/7 > 1/2.
\end{equation}
If this happens, then there is some subset $X \subseteq N_U(u)$ of size at least $d-c$ contained in an $\eta$-good palette $R$. In some branch we identify both $X$ and $R$, and assign $X$ to $W_R$.

Suppose the algorithm always successfully identifies such $u$'s, $X$'s, and $R$'s until reaching base case 1. Then all vertices in the large set $W_R$ correctly have hidden color in $R$, and \Cref{lem:forcing} will return yes. If we reach base case 2, then some coloring of $W$ matches the witness coloring, and the resulting list color problem has a solution. (In fact it will accept whenever $G$ is $k$-colorable, which is fine.) The only source of randomness on this path is over the choice of $u$'s. If we perform the main step $t$ times before reaching a base case, then the probability of only choosing good $u$'s is at least $2^{-t}$ by the calculation of \Cref{eqn:uprob}. But from the criterion of base case 1, the main step can be executed at most $\lceil 2^k \lambda n/r \rceil$ times, so this probability of accepting is greater than $(1/2)^{\lceil 2^k \lambda n/r \rceil} = C_{rep}^{-1}$. Thus the final success probability after all trials is at least $1-(1-C_{rep}^{-1})^{C_{rep}} > 1-e^{-1}$.

Finally, suppose $G$ is not $k$-colorable. Then case 1 will always reject, by correctness of \Cref{lem:forcing}. Similarly, case 2 can only accept if $G$ is $k$-colorable. Note that in Case 2 we may accept if $G$ is $k$-colorable with a profile other than $p$, which again is permitted.
\paragraph{Runtime}
We spend $O^*(2^{H(1/2 - \eta)n})$ time for the preprocessing of \Cref{lem:forcing}.

For one call to the subroutine, consider the computation tree where internal nodes are main steps, and leaves are base cases. Each main step has branching factor at most 
\[B := 2^k \cdot \sum_{i=0}^c \binom{d}{i} \]
and this tree has depth at most $\lceil 2^k \lambda n/r \rceil$, so it has at most $B^{\lceil 2^k \lambda n/r \rceil}$ leaves. Over all repetitions, there are therefore at most \[B^{\lceil 2^k \lambda n/r \rceil} \cdot C_{rep} = (2B)^{\lceil 2^k \lambda n/r \rceil} \le O(2^{\beta 2^k \lambda n}) \]
where $\beta = \log(2B)/r$, where the final inequality holds because $B$ is a constant.

A leaf in base case 1 costs $2^{n-|W_R|} < 2^{n(1-\lambda)}$, so the total cost from all leaves of this type is bounded by
\begin{equation}\label{case1}
	O^*(2^{\beta  2^k \lambda n} \cdot 2^{n(1-\lambda)}) \le O^*(2^{(1-\lambda + \beta 2^k \lambda)n}).
\end{equation}
(Note that the cost of an internal node, i.e.~a main step, is $\poly(n)$.)

A leaf in case 2 costs 
\[k^w \cdot 2^{(1 - \mu)(n-w)} = 2^{(1 - \mu)n + (\log k - 1 + \mu)w} \le 2^{(1 - \mu)n + (\log k - 1 + \mu)2^k \lambda n} ,\]
so multiplying this by the leaf bound gives the cost
\begin{equation}\label{case2}
 O^*(2^{(1 - \mu + 2^k \lambda (\log k - 1 + \mu + \beta))n}).
\end{equation}

From \Cref{case1}, if $2^k \beta < 1/2$ the cost of leaves of type 1 is at most $2^{(1-\lambda / 2)n}$. So it suffices if

\[1/2 > 2^k(\log 2B)/r = 2^k(1 + k + \log \left (\sum_{i=0}^c \binom{d}{i}\right ))/(d-c).\]
This clearly holds for $d(k,c)$ large enough. Concretely, we may take $d = 2^{k+4}(c+1)(k+c+4)$: since $d > 2c$, using that $\sum_{i=0}^c \binom{d}{i} < (ed/c)^c$ and $d/c < 2^{k+5}(k+c+4)$, the above is at most
\[2^k \frac{k + c \log (ed/c) + 1}{2^{k+4}(c+1)(k+c+4) - c} \le \frac{k + c \log (4 \cdot 2^{k+5}(k+c+4)) + 1}{8 (c+1)(k+c+4)} \le \frac{k + c(k+7 + (k + c + 4) ) + 1}{8(c+1)(k+c+4)} < 1/2.\]

From \Cref{case2}, if
\[2^k \lambda (\log k - 1 + \mu + \beta) < \mu/2\]
which follows from our choice of
\[\lambda \le \frac{\mu}{2^{k+1} (\log(2^{k+1} \cdot \sum_{i=0}^c \binom{d}{i})/(d-c) + \log k + \mu - 1)}, \]
then this case has complexity at most $2^{(1-\mu/2)n}$.

So the final runtime is bounded by $O^*(2^{H(1/2-\eta)n} + 2^{(1-\mu/2)n} + 2^{(1-\lambda/2)n})$, with $\mu > 0, \lambda > 0$ being functions of $k$ and $c$.
\end{proof}

\begin{lemma}\label{lem:averaging}
	Let $\eta \le 1/(6(k-1))$. Let $p$ be a distribution on $[k]$, let $\emptyset \neq D \subset [k]$ with $0 < p(D) < 1/2 + \eta$, and let $u$ be a distribution on $D$. Then either there exists $i,j \in D$ such that $u(\{i, j\}) > 1/2 + 1/6k$, or there exists $i \in D$ such that $p(D \setminus \{i\}) < 1/2 - \eta$ and $u(D \setminus \{i\}) > 1/2 + 1/(3k)$. 
\end{lemma}
\begin{proof}
First suppose that for some $i$, $u(i) \ge 1/2-c$ with $c = 1/(3k)$. If $|D| = 1$ then in fact $u(i) = 1$ and we are done (with $j=i$). Otherwise, by averaging, there is a $j \neq i \in D$ with $u(j) \ge (1-u(i))/k$. Together, \[u(\{i,j\}) = u(i) + u(j) > 1/2 - c + (1/2 + c)/k = 1/2 + 1/2k + c/k -c\]
which is greater than $1/2 + 1/6k$. 
	
Otherwise suppose that $u(i) < 1/2 - c$ for all $i \in D$. Then  $u(D \setminus \{i\}) > 1/2 + c$ for all $i$. Furthermore, for some $i \in D$,  $p(i) > p(D)/k$. Then 
\[p(D \setminus \{i\}) < p(D) - p(D)/k = p(D)(1-1/k) < (1/2 + \eta)(1-1/k)\]
so if $\eta \le (1/2 - 1/(3k))/(1-1/k) - 1/2 = 1/(6(k-1))$ this is at most $1/2 - 1/(3k) \le 1/2 - \eta$.
\end{proof}

\begin{theorem}\label{thm:main}
For all $k \ge 3$, there exists $\varepsilon_k > 0$ such that, there is a randomized algorithm running in time $(2-\varepsilon_k)^n$ which, given as input an $n$-vertex graph $G$ and profile $p$,
\begin{itemize}
	\item Returns ``yes" with constant probability if $G$ admits a $k$-coloring with profile $p$,
	\item Returns no if $G$ is not $k$-colorable.
\end{itemize}
\end{theorem}

\begin{proof}
Throughout, we assume that $n\geq n_0(k)$ for a sufficiently large constant
$n_0(k)$.
\paragraph{Parameters} We fix parameters
\begin{align*}
\varepsilon &= 1/100,\\
\eta &= \frac{1}{6(k-1)},\\
\zeta & = -\log 1/(2-1/(3k)),\\
T &= \bigg \lceil \frac{2 \ln 4k  \cdot \log k \cdot 2^{2k + 7} \cdot 6k}{1-\zeta} \bigg \rceil,\\
c &= kT.
\end{align*}
Finally, let $d(k,c)$ denote the degree threshold supplied
 by \Cref{lem:degen}, and set
 \[
 d=\max\{T(k-1)+1,d(k,c)\}.
 \]
Note that these are ultimately all functions of $k$.



\begingroup
\RestyleAlgo{boxed}
\begin{algorithm}[t]
	\DontPrintSemicolon
	\SetAlgoVlined
	\SetInd{0.5em}{1.25em}
	
	\TitleOfAlgo{A fixed-profile $k$-coloring algorithm.}
	\label{alg:coloring}
	
	\If{at least $\varepsilon n$ vertices have degree at most $d$}{
		Run \Cref{thm:zamir} and return its result.\;
	}

	\BlankLine
	
	Run \Cref{lem:degen} with $c=kT$.\;
	\If{it accepts}{
		Accept.\;
	}
	
	\BlankLine
	Perform the preprocessing step of \Cref{lem:forcing}.\;
	
	Sample a uniformly random $S \subset V(G)$ of size $\lceil \sigma n \rceil$, where
	$\sigma=\ln(4k)/T$.\;
	\BlankLine
	\ForEach{valid $k$-coloring of $S$}{
		For $R \subseteq [k]$, let $V_R$ be the set of vertices in $G$ incident to a vertex in $S$ of each color
		in $R$, and to no vertex in $S$ of color not in $R$.\;
		
		\BlankLine
		
		\ForEach{$R \subset [k]$ such that
			$|V_R|>2^{-2k-6} \cdot n$}{
			
			Let $a_R \gets \lfloor \frac{|V_R|}{6k} \rfloor$.\;
			
			\ForEach{$i,j\in[k]$}{
				\For{$t \gets 1$ \KwTo $\lfloor 2^{\zeta \cdot a_R} \rfloor$}{
					Sample a uniformly random $A\subseteq V_R$ of size
					$a_R$.\;
					
					Assign $A$ to palette $\{i,j\}$ and run
					\Cref{lem:forcing}.\;
					
					\If{\Cref{lem:forcing} accepts}{
						Accept.\;
					}
				}
			}
			
			\BlankLine
			
			\ForEach{$i \in [k]$ such that $R^c\setminus\{i\}$ is nonempty and $\eta$-good}{
				\For{$t \gets 1$ \KwTo $\lfloor 2^{\zeta\cdot a_R} \rfloor$}{
					Sample a uniformly random $A\subseteq V_R$ of size
					$a_R$.\;
					
					Assign $A$ to palette $R^c\setminus\{i\}$ and run
					\Cref{lem:forcing}.\;
					
					\If{\Cref{lem:forcing} accepts}{
						Accept.\;
					}
				}
			}
		}
	}
	
	\BlankLine
	
	Reject.\;
\end{algorithm}
\endgroup

\paragraph{Correctness}
Suppose that $G$ is not $k$-colorable. Because our algorithm only accepts after either calling \Cref{thm:zamir}, \Cref{lem:degen}, or \Cref{lem:forcing}, it follows from the one-sidedness of those algorithms that we will reject.

Now suppose $G$ is $k$-colorable with profile $p$. For the rest of the proof, we fix such a witness coloring $C = (C_1, \ldots, C_k)$.

First, if there are at least $\varepsilon n$ vertices of degree at most $d$, the algorithm just runs \Cref{thm:zamir}, so it is correct. Note that it may accept even if $G$ is $k$-colorable but not with profile $p$. From now on suppose that fewer than $\varepsilon n$ vertices have degree at most $d$.

For a vertex $v$, let 
\[\Sigma_T(v) = \{i \in [k]: |N(v) \cap C_i| \ge T\}\]
and let
\[U_{R} = \{v \in V(G) : \Sigma_T(v) = R\}.\]
Since $d \ge T(k-1) + 1$, every vertex of degree at least $d$ is contained in some $U_{R}$, so only the at most $\varepsilon n$ low-degree vertices are absent from any $U_{R}$ with $R \neq \emptyset$. Note that $U_{[k]} = \emptyset$, as no vertex can see its own color in a valid coloring.

Suppose the witness coloring satisfies the condition that at least a $15/16$ fraction of vertices have degree at least $d$ and belong to some $U_{R}$ with $R$ being $\eta$-good. Then \Cref{lem:degen}, ran with $c = Tk$, will detect this coloring with constant probability. Here we're using the fact that if $v$ is in $U_{R}$, it has fewer than $kT$ ``exceptional" neighbors with color not in $R$. Also, $d$ was chosen to satisfy the condition of \Cref{lem:degen}.

Suppose the witness coloring does not satisfy this condition. Recall that at least a $99/100$ fraction of vertices have degree at least $d$ (and hence also belong to some $U_{R}$). So if less than a $15/16$ fraction of vertices have degree at least $d$ and belong to a $U_{R}$ with $R$ good, then at least a $99/100 - 15/16 > 2^{-5}$ fraction of vertices must belong to
\[\bigcup_{p(R) > 1/2 - \eta} U_{R}.\]
By averaging, there exists an $R_0$ such that $p(R_0) > 1/2 - \eta$ and $|U_{R_0}|/n > 2^{-k-5} =: \gamma$. We now argue that some $V_{R}$ with $R_0 \subset R$ is also large with constant probability.

Let $S \subset V(G)$ with $|S| = \lceil \sigma n \rceil = \lceil \ln(4k) n/T \rceil$ be the randomly chosen set. For $v \in U_{R_0}$, the probability that $v$ does not see one of its neighbors of color $i \in R_0$, of which there are at least $T$, is at most $(1-\ln(4k)/T)^T < e^{-\ln(4k)} = 1/(4k)$. Hence the probability that $v$ sees all colors in $R_0$ is at least
\[1-|R_0| \cdot 1/(4k) > 3/4\]
Thus the expected number of $v \in U_{R_0}$ seeing one of each color in $R_0$ is at least
\[3|U_{R_0}|/4 \]
and since this random variable is at most $|U_{R_0}|$, it is at least $|U_{R_0}|/2$ with probability at least $1/2$.
If this happens, there then exists some $V_{R}$ with $R_0 \subseteq R$ and with size at least
\begin{equation}\label{eqn:bigVR}
|U_{R_0}|/2
\cdot 2^{-k} \ge \gamma\cdot 2^{-k-1} \cdot n = 2^{-2k-6} \cdot n .
\end{equation}
The above algorithm identifies this set with constant probability.

From now on, consider the branch where there exists such a $V_R$, we identified $R$, and we correctly guessed the witness coloring restricted to $S$.
 
Let $u_R$ be the distribution on $[k]$ where $u_R(i)$ is proportional to the number of vertices in $V_R$ whose witness color is $i$. Note $u_R$ is supported on $R^c$, and also note that $1/2 - \eta < p(R_0) \le p(R)$, so $p(R^c) < 1/2 + \eta$. Applying \Cref{lem:averaging} with $D = R^c$ and $u=u_R$, at least one of the following holds:
\begin{itemize}
	\item There exist $i,j \in R^c$ with $u_{R}(\{i,j\}) > 1/2 + 1/(6k)$. In this case, consider the branch where we correctly guess $i,j$. We now analyze the probability of success after choosing $A$ randomly and invoking \Cref{lem:forcing}. Recall $|A| = \lfloor |V_R|/(6k) \rfloor$.

	The probability that $A$ is contained in $V_R \cap (C_i \cup C_j) := Q$ is at least
	\begin{align*}
	\frac{\binom{|Q|}{|A|}}{\binom{|V_R|}{|A|}} & = \prod_{\ell=0}^{|A|-1} \frac{|Q|-\ell}{|V_R|-\ell} \ge \left ( \frac{|Q|-|A|}{|V_R| - |A|} \right )^{|A|} \ge   \left ( \frac{(1/2 + 1/(6k))|V_R|- |A|}{|V_R| - |A|} \right )^{|A|}\\
	&=    \left ( \frac{(1/2 + 1/(6k))- |A/|V_R|}{1 - |A|/|V_R|} \right )^{|A|} \ge \left ( \frac{(1/2 + 1/(6k)) - 1/(6k)}{1-1/(6k)} \right )^{|A|} = 2^{-\zeta |A|}
	\end{align*}

	where the last inequality follows since $(c-x)/(1-x)$ is decreasing for $x \in (0,1)$ when $c \in (0,1)$. Recall $\zeta = -\log 1/(2-1/(3k))$. Repeating $2^{\zeta|A|}$ many times, with constant probability it holds that $A \subset Q$ for some trial. Once this happens, \Cref{lem:forcing} will detect the hidden coloring.
	\item Suppose there exists $i \in R^c$ with $p(R^c \setminus \{i\}) < 1/2 - \eta$ and $u_{R}(R^c \setminus \{i\}) > 1/2 + 1/(3k) > 1/2 + 1/(6k)$. 
	Assign a random subset $A \subset V_{R}$ of size $|V_R|/(6k)$ to the palette $R^c \setminus \{i\}$, and run \Cref{lem:forcing}. From the analysis in the previous case, by repeating $2^{\zeta|A|}$ times, with constant probability $A \subseteq \cup_{j \in R^c \setminus \{i\}} C_j$ for some trial. Once that happens, \Cref{lem:forcing} detects the hidden coloring.
\end{itemize}

\paragraph{Runtime}
The parameters $\eta, d, c$ are all functions of $k$, and $\varepsilon$ is an absolute constant. It follows that the call to \Cref{thm:zamir}, the preprocessing step of \Cref{lem:forcing}, and the call to \Cref{lem:degen} all run in time $O((2-f(k))^n)$ for some $f(k)>0$.

We then enumerate all $k^{\sigma n}$ colorings of $S$. Note that the number of $R$ with $|V_R|$ large, the number of pairs $i,j \in [k]$, and the number of $i$ with $R^c \setminus \{i\}$ being $\eta$-good, are all functions of $k$, and hence constant. We then sample $\lfloor 2^{\zeta|A|} \rfloor$ random $A$'s and run \Cref{lem:forcing}, which takes time $O^*(2^{n - |A|})$. The runtime of this stage (lines 9--23) is therefore bounded by
\[O^*(k^{\sigma n} \cdot 2^{\zeta|A|} \cdot 2^{n-|A|}),\]
which has exponential rate
\[1+\sigma \log k + |A|/n(\zeta - 1) = 1+\ln(4k)/T \log k + |A|/n(\zeta - 1).\]
Since 
\[|A| = \lfloor |V_R|/(6k) \rfloor > |V_R|/(12k) \ge n/(2^{2k + 7} \cdot 6k)\]
by \Cref{eqn:bigVR}, and since $\zeta - 1 < 0$,
this rate is at most
\[1+ (\ln 4k  \cdot \log k)/T  + (\zeta - 1)/(2^{2k + 7} \cdot 6k).  \]
For sufficiently large $T$ depending on $k$, this is less than one. For instance, taking
\[T \ge \frac{2  (\ln 4k  \cdot \log k \cdot 2^{2k + 7} \cdot 6k)}{1-\zeta}\]
this stage takes time
\[O^*(2^{n(1+ (\zeta-1)/(2^{2k + 8} \cdot 6k))}). \qedhere \]
\end{proof}

By enumerating all profiles and running \Cref{thm:main} for each, and then repeating this constantly many times to amplify the success probability to $2/3$, we conclude \Cref{thm:fin}.

\section{AI Disclosure}
The author used OpenAI’s ChatGPT (GPT-5.5 and 5.6 Pro) for technical drafting, runtime and parameter analysis, and exploring proof strategies. Most significantly, it contributed to the proof of \Cref{lem:degen} by repairing an initial algorithm of the author's by suggesting the second base case. The author independently verified all claims appearing in the manuscript, and takes sole responsibility for the work.
\bibliographystyle{amsalpha}
\bibliography{refs}
\end{document}